\documentclass[a4paper,11pt]{article} 

\usepackage[top=2.5cm, bottom=2.5cm, left=2.5cm, right=2.5cm]{geometry} %%size

\usepackage[english]{babel} %%language 
  
\usepackage[utf8]{inputenc} %%entering special characters

\usepackage[T1]{fontenc} %%display special characters

\usepackage{amssymb,amsmath} %%math characters

\usepackage{amsthm} %%to create newtheorem

\usepackage{mathtools} %%extension to amsmath

\usepackage{tikz} %%tikz figures

\usepackage{hyperref} %%including references and external links

\usepackage{lmodern} %%%more latin characters

\usepackage{microtype} %%adjustement

\usepackage{array}  %%%tables

\usepackage{float} %%for floating object (tables, figures...)

\usepackage{indentfirst} %%Indent first paragraph after section header

\usepackage[ruled,lined]{algorithm2e} %%type of displaying algorithms, lined for begin and end

\usepackage[justification=centering]{caption}  %%to write captions

\newcolumntype{P}[1]{>{\centering\arraybackslash}p{#1}} %%centering text in tables

\providecommand{\keywords}[1]{\textbf{\textit{Keywords:}} #1} %%keywords

\SetKw{KwBy}{by} %%For loop in algorithm2e
\newtheorem{theorem}{Theorem}[section]
\newtheorem{cor}{Corollary}[section]
\newtheorem{lemma}{Lemma}[section]
\newtheorem{prop}{Proposition}[section]

\newcommand{\cube}[3]{
			\fill[color = black!20] (#1,#2,#3) +(0,1,0) -- +(0,1,-1)
 -- +(1,1,-1) -- +(1,1,0) -- cycle;
			\fill[color = black!30] (#1,#2,#3) -- +(0,1,0) -- +(1,1,
0) -- +(1,0,0) -- cycle;
			\fill[color = black!50] (#1,#2,#3) +(1,0,0) -- +(1,1,0) 
-- +(1,1,-1) -- +(1,0,-1) -- cycle;
			\draw[color = black!80] (#1,#2,#3) -- +(0,1) -- +(1,1) -
- +(1,0) -- cycle;
.025,.94,0) -- +(0,1,-1) -- +(1,1,-1) -- +(1,0,-1) -- +(1,0,0);
			\draw[color = black!80] (#1,#2,#3) + (0,1,0) -- +(0,1,-1
) -- +(1,1,-1) -- +(1,0,-1) -- +(1,0,0);
			\draw[color = black!80] (#1,#2,#3) +(1,1,0) -- +(1,1,-1)
;
}

\newcommand{\cell}[2]{
			
			\fill[color = black!30] (#1,#2) -- +(0,1) -- +(1,1) -- +(1,0) -- cycle;
					\draw[color = black!80] (#1,#2) -- +(0,1) -- +(1,1) -- +(1,0) -- cycle;
;
}
\newcommand{\freeza}[3]{
			
			\fill[color = black!30] (#1,#2,#3) -- +(0,1,0) -- +(1,1,0) -- +(1,0,0) -- cycle;
					\draw[color = black!80] (#1,#2,#3) -- +(0,1,0) -- +(1,1,0) -- +(1,0,0) -- cycle;
;
}
\newcommand{\buu}[3]{
			
			\fill[color = black!30] (#1,#2,#3) -- +(0,0,1) -- +(1,0,1) -- +(1,0,0) -- cycle;
					\draw[color = black!80](#1,#2,#3) -- +(0,0,1) -- +(1,0,1) -- +(1,0,0) -- cycle;
;
}
\title{\textbf{Enumeration of parallelogram polycubes}}
\author{Abderrahim Arabi\\
\small USTHB, Faculty of Mathematics\\
\small RECITS Laboratory\\
\small BP 32, El Alia 16111, Bab Ezzouar\\
\small Algiers, Algeria\\
\small\tt arabi.abderrahim@gmail.com\\
\small\tt rarabi@usthb.dz\\
\and
Hacène Belbachir\\
\small USTHB, Faculty of Mathematics\\
\small RECITS Laboratory\\
\small BP 32, El Alia 16111, Bab Ezzouar\\
\small Algiers, Algeria\\
\small\tt hacenebelbachir@gmail.com\\ 
\and
Jean-Philippe Dubernard\\
\small University of Rouen-Normandie, Faculty of Science and Technique, LITIS\\[-0.8ex]
\small Avenue de l’université 76800 Saint-Étienne-du-Rouvray\\[-0.8ex]
\small Rouen, France\\
\small\tt jean-philippe.dubernard@univ-rouen.fr
}
\date{}

\begin{document}

\maketitle

\begin{abstract}
In this paper, we enumerate parallelogram polycubes according to several parameters. After establishing a relation between Multiple Zeta Function and the Dirichlet generating function of parallelogram polyominoes, we generalize it to the case of parallelogram polycubes. We also give an explicit formula and an ordinary generating function  of parallelogram polycubes according to the width, length and depth, by characterizing its projections. Then, these results are generalized to polyhypercubes.  
\end{abstract}


\section{Introduction}\label{Sect0}
In $\mathbb{Z}^2$, a polyomino is a finite connected union of cells without a cut point and defined up to translation \cite{klarner1965some}. An open problem in combinatorics, is finding the number of   
polyominoes for a giving a number $n$ of cells. It is considered as a hard problem and, to date, no exact formula is known. Polyominoes have been enumerated by algorithms and the values are known up to $n=56$ \cite{jensen2003counting}. In the absence of formulas for the general case and to approximate their enumeration, polyominoes with special proprieties were defined and enumerated in the literature. These families of polyominoes have convex, directed or others constraints. We can cite the column convex, the convex, the diagonally convex and the directed. Exact enumerations exist according to the number of cells for some of them and others were enumerated according to additional parameters as the perimeter, the height and the width. One can find a survey in \cite{guttmann2009polygons}. Also methods of enumeration were used as Temperley Methodology \cite{temperley1956combinatorial}, Bousquet Melou Method \cite{bousquet1996method} and ECO Method \cite{doi:10.1080/10236199908808200}. One particular family studied is the family of \textit{Parallelogram Polyominoes}, the polyominoes of this class are columns without holes glued together with the bottoms and the tops of the columns increasing. Exact enumeration exists for it according to different parameters. Let us cite for instance the results of Delest and Viennot\cite{DELEST1984169} according to the perimeter using bijection with Dyck paths, Delest, Dubernard and Dutour \cite{DELEST1995503} according to the area, width, right and left corners and Bousquet-Melou \cite{bousquet1996method} according to the area, width, height and length of the leftmost and rightmost column.      

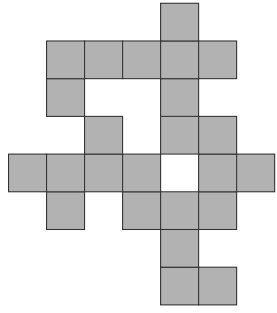
\begin{figure}[H]
\centering
\begin{tikzpicture}[scale=0.5]

	\foreach \x/\y in {1/4,2/3,2/4,2/6,2/7,3/4,3/5,3/7,4/3,4/4,4/7,5/1,5/2,5/3,5/5,5/6,5/7,5/8,6/1,6/3,6/4,6/5,6/7,7/4}{
					\cell{\x}{\y};
				}

\end{tikzpicture}
\label{f1i}
\caption{Example of a polyomino.}
\end{figure}

The extension of polyominoes in dimension $3$ are called polycubes. In $\mathbb{Z}^3$, a unit cell is defined as unit cube. A polycube is a finite face-connected union of elementary cells defined up to translation. As the $2$-dimensional case the enumeration of polycubes with $n$ cells is still an open problem. Many authors have enumerated the first values of polycubes. In $1971$, Lunnon enumerated them up to $7$ \cite{lunnon1971counting}. In $2008$ Aleksandrowicz and Barequet gave the bound up to $18$ \cite{doi:10.1142/S0218195909002927} and the known upper bound is from Luther and Mertens up to $19$,  \cite{luther2011counting}. Unlike polyominoes, only a few classes of polycubes have been enumerated. Let us cite for instance, the plane partitions \cite{cohn1998shape}, the directed plateau polycubes \cite{champarnaud2013enumeration}, and the partially directed snake polycubes \cite{GOUPIL2018223}. Some tools were developed for the enumeration of polycubes, in particular the generic method \cite{jeanne2013generic}, an extension of Bousquet-Melou method \cite{bousquet1996method} and the Dirichlet convolution for the enumeration of polycubes \cite{carre2015enumeration}.\\ 
These methods only enumerate directed polycubes with convex restrictions. The problem for other classes is still open and no method of $2$-dimensional case could have been adapted to the $3$-dimensional case. The class of parallelogram polycubes has been investigated. In \cite{champarnaud2013enumeration}, the authors found a differential equation for the generating function according the volume, width, the area of the rightmost face, the height of the last plateau and the depth of the last plateau, but this equation could not be solved. However, some asymptotic results, that are the only known, are given. Also, the two previous methods did not work for this class.

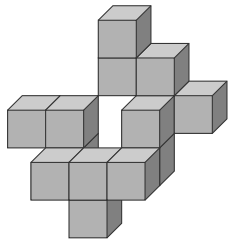
\begin{figure}[h]
\centering
\begin{tikzpicture}[scale=0.5]
	\foreach \x/\y/\z in {1/1/1,1/2/1,-2/2/1,-1/1/1,-1/2/1,2/2/0,0/3/0,1/3/0,0/4/0,0/0/2,-1/1/2,0/1/2,1/1/2						  
						   }{
					\cube{\x}{\y}{\z};
				}					
\end{tikzpicture}
\caption{Example of a polycube.} \label{f2i}
\end{figure}
%%%%
%%In this paper we study the family of parallelogram polycubes. In \cite{champarnaud2013enumeration}, the authors found a differential equation for the generating function according the volume, width, the area of the rightmost face, the height of the last plateau and the depth of the last plateau. And some asymptotic were given, in this paper after giving definitions in Section \ref{Sect02}. We use another approach to enumerate parallelogram polycubes, using the Dirichlet generating function in Section \ref{Sect03}, we start by enumerating parallelogram polyominoes, then the case of polycubes is deduced from it. In Section \ref{Sect04}, we use the projection relation between parallelogram polyominoes and polycubes to enumerate the last ones according to the width, height and depth. And finally, in Section \ref{Sect05}, we generalize some results to any dimension $d\geq 4$. 

Our objective in this paper is to enumerate the family of parallelogram polycubes in two different ways by new approaches. The first uses Dirichlet generating function to enumerate them according to the width and volume. It leads us to enumerate parallelogram polyominoes. Then we generalize it to the polycube case. We also show the relation between this generating functions and the Multiple Zeta Function \cite{hoffman1992multiple}. The second enumeration is done according to the width, height and depth. In this case we project the polycubes and using the known results for polyominoes, we deduce an explicit formula for polycubes. In Section \ref{Sect02}, we give some definitions and notations. Then in Section \ref{Sect03}, we enumerate them according to the width and volume. In Section \ref{Sect04}, we explore the parallelogram polycubes according to the width, height and depth. Finally, in the last Section, we generalize some results to any dimension $d\geq 4$.

\section{Preliminaries}\label{Sect02}

Let $(0,\vec{i},\vec{j})$ be an orthonormal coordinate. 
%We define several parameters for a polyomino: 
The area of a polyomino is the number of its cells, its width is the number of its columns and its height is the number of its rows.\\
A polyomino is said to be \textit{column-convex} (resp. \textit{row-convex}) if its intersection with any vertical (resp. horizontal) line is connected. %And i
If it is both column and row convex, it is called \textit{convex} polyomino.\\ 
A North (resp. East) step is a movement of one unit in $\vec{i}$-direction (resp. $\vec{j}$-direction). From this two steps, a \textit{directed} polyomino is defined as if from a distinguished cell called \textit{root}, we can reach any other cell by a path that uses only North or East steps.\\
The \textit{bottom} (resp. \textit{top}) of a column the height of the closet (resp. furthest) cell to the axis  $(0,\vec{i})$.\\
A \textit{parallelogram polyomino} is defined as a convex polyomino such that the bottoms and the tops of its columns form two increasing sequences, see Fig. \ref{f12}.

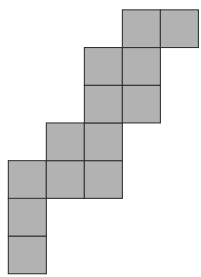
\begin{figure}[H]
\centering
\begin{tikzpicture}[scale=0.5]

	\foreach \x/\y in {0/0,0/1,0/2,1/2,1/3,2/2,2/3,2/4,2/5,3/4,3/5,3/6,4/6}{
					\cell{\x}{\y};
				}

\end{tikzpicture}
\caption{A parallelogram polyomino.}
\label{f12}
\end{figure}

Let $(0,\vec{i},\vec{j},\vec{k})$ be an orthonormal coordinate system. As for polyominoes, several parameters can be defined for a polycube. The \textit{volume} is the number of its cubes.
The \textit{width} (resp. \textit{height}, \textit{depth}) of a polycube is the difference between its greatest and its smallest indices according to $\vec{i}$ (resp. $\vec{j}$, $\vec{k}$). \\
A polycube is said to be \textit{directed} if each of its cells can be reached from a distinguished cell, called the \textit{root}, by a path only made of East (one unit in the $\vec{i}$-direction), North (one unit in $\vec{j}$-direction) and Ahead (one unit in $\vec{k}$-direction) steps.\\
The front (resp. the back) of a plateau as the closest (resp. the furthest) side of the plane $(0,\vec{i},\vec{j})$. And the bottom (resp. the top) of a plateau are defined as the closet (resp. the furthest) side of the plane $(0,\vec{i},\vec{k})$.\\
A \textit{stratum} is a polycube of width $1$ and a \textit{plateau} is a rectangular stratum with no holes, it is the equivalence of a column for a polycube. \\
A \textit{plateau polycube} is a polycube whose strata are plateaus glued together.
A subclass of plateau polycubes are \textit{Parallelogram Polycubes}. A parallelogram polycube is a plateau polycube such that the bottoms, the tops, the fronts and the backs of its plateaus form an increasing sequences, see Fig. \ref{f22}.

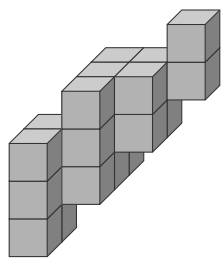
\begin{figure}[H]
\centering
\begin{tikzpicture}[scale=0.5]
	\foreach \x/\y/\z in {0/0/0,0/1/0,0/2/0,0/0/1,0/1/1,0/2/1,
						  1/1/-2,1/2/-2,1/3/-2,1/1/-1,1/2/-1,1/3/-1,1/1/0,1/2/0,1/3/0,
						  2/2/-2,2/3/-2,2/2/-1,2/3/-1,
						  3/3/-2,3/4/-2						  
						   }{
					\cube{\x}{\y}{\z};
				}					
\end{tikzpicture}
\caption{Example of a parallelogram polycube.}
\label{f22}
\end{figure}
\noindent
For more details about polyominoes and polycubes one can see \cite{champarnaud2013enumeration}.
\\
We define $S(X)$ of a finite set $X$ as the set of all the ordered partition of the set $X$. For example if $X=\{a,b,c\}$ then $S(X)=\{(\{a\},\{b\},\{c\}),
(\{a\},\{c\},\{b\}),
(\{b\},\{a\},\{c\}),
(\{b\},\{c\},\{a\}),\\
(\{c\},\{a\},\{b\}),
(\{c\},\{b\},\{a\}),
(\{a,b\},\{c\}),
(\{a,c\},\{b\}),
(\{b,c\},\{a\}),
(\{a\},\{b,c\}),
(\{b\},\{a,c\}),\\
(\{c\},\{a,b\}),
(\{a,b,c\})
\}.$\\
\noindent
Given a sequence $\{a_{n_1,n_2,...,n_k}\}_{n_1,n_2,...,n_k\geq 1}$, the \textit{ordinary generating function} OGF \cite{wilf2013generating} is the formal power series of the form: 
$$F(x_1,x_2,...,x_k)=\sum_{n_1,n_2,...,n_k\geq 1}a_{n_1,n_2,...,n_k}x^{n_1}x^{n_2}\cdots x^{n_k}.$$
For this sequence, the \textit{Dirichlet generating function} DGF \cite{wilf2013generating} is the formal power series of the form:
$$V(x_1,x_2,...,x_k)=\sum_{n_1,n_2,...,n_k\geq 1}\frac{a_{n_1,n_2,...,n_k}}{n_1^{x_1}n_2^{x_2}\cdots n_k^{x_k}}.$$
For $k\geq 1$, the \textit{Multiple Zeta Function} \cite{hoffman1992multiple} is defined as:
$$\zeta_k(x_1,x_2,...,x_k)=\sum_{n_1>n_2>\cdots >n_k\geq 1}\frac{1}{n_1^{x_1}n_2^{x_2}\cdots n_k^{x_k}}.$$
Note that the case $k=1$ is the classic Riemann Zeta Function \cite{wilf2013generating}, defined as 
$$\zeta_1(x)=\sum_{n>1}\frac{1}{n^x}.$$ 
%the case $k=1$ is the classic Riemann Zeta Function.

\section{Enumeration of parallelogram polycubes according to the width and volume}\label{Sect03}
In this section, we give an expression of the Dirichlet generating function of parallelogram polyominoes. Then we deduce a relation between this case and the case of polycubes.\\
We start by giving formulas for the number of parallelogram polyominoes according to the width and the area and their extension for polycubes according to the width and volume.

\subsection{Formulas for parallelogram polyominoes and polycbes}   
Let $a_{m_1,m_2,...,m_k}$ be the number of parallelogram polyominoes of width $k$ and whose area (or height) of the $i^{th}$ column is equal $m_i$, where $1 \leq i \leq k$. The following proposition is a consequence from the definition. However, as we use a similar reasoning in the $3$-dimensional case, we give its proof.
%we give a proof, because we use it in the $3$-dimensional case.

\begin{prop}\label{p13}
For integer $k$ and $m_1,m_2,...,m_k\in \mathbb{N}$, we have:
\[a_{m_1,m_2,...,m_k} = \left\{ 
\begin{array}{l l}
  1, & \quad \text{if $k=1,$}\\
  \prod_{j=1}^{k-1}min(m_j,m_{j+1}), & \quad \text{otherwise.}\\ \end{array} \right. \]
\end{prop}

\begin{proof}
In the case $k=1$, the considered polyominoes (that are of width $1$) are reduced to one column. Thus, for any value of $m_1$, there is only one possible polyomino. To determinate $a_{m_1,...,m_k}$ in the case $k \geq 2$, we have to build all possible corresponding polyominoes. We start by considering a column of height $m_1$ and we successively glue all the other columns one by one. When we add the $i^{th}$ column of height $m_i$ onto the $(i+1)^{th}$ one of height $m_{i+1}$ $(1 \leq i \leq k-1)$, there is exactly $min (m_i, m_{i+1})$ possibilities to obtain a parallelogram polyomino. An illustration of this building is given Fig. \ref{f13}.
\end{proof}

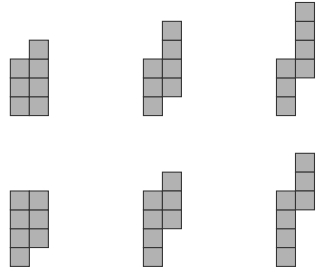
\begin{figure}[H]
\centering
\begin{tikzpicture}[scale=0.25]

	\foreach \x/\y in {0/0,0/1,0/2,1/0,1/1,1/2,1/3,7/0,7/1,7/2,8/1,8/2,8/3,8/4,14/0,14/1,14/2,15/2,15/3,15/4,15/5}{
					\cell{\x}{\y};
				}
				
				\foreach \x/\y in {0/-8,0/-7,0/-6,0/-5,1/-7,1/-6,1/-5,7/-8,7/-7,7/-6,7/-5,8/-4,8/-6,8/-5,14/-8,14/-7,14/-6,14/-5,15/-3,15/-4,15/-5}{
					\cell{\x}{\y};
				
				}
				
\end{tikzpicture}
\caption{Example of construction of different gluing to obtain a parallelogram polyomino.}
\label{f13}
\end{figure}

Note that, summing for all ordered partition of an integer $n$ into $k$ parts, we get $b_{k,n}$ the number of parallelogram polyominoes of width $k$ and area $n$ in the following corollary.

\begin{cor}
For integers $n\geq k \geq 1$,
$$b_{k,n}=\sum_{\substack{m_1,m_2,...,m_k\in \mathbb{N}\\m_1+m_2+...+m_k=n}}a_{m_1,m_2,...,m_k}.$$
 
\end{cor}

The first values of $b_{k,n}$ are given in the following table.
\begin{table}[H]
\begin{tabular}{|c|P{1cm}P{1cm}P{1cm}P{1cm}P{1cm}P{1cm}P{1cm}P{1cm}P{1cm}P{1cm}|}\hline
$n \backslash k$  & 1 & 2& 3&4&5&6&7&8&9&10 \\ \hline
1&1&0&0&0&0&0&0&0&0&0\\
2&1&1&0&0&0&0&0&0&0&0\\
3&1&2&1&0&0&0&0&0&0&0\\
4&1&4&3&1&0&0&0&0&0&0\\
5&1&6&8&4&1&0&0&0&0&0\\
6&1&9&17&13&5&1&0&0&0&0\\
7&1&12&32&34&19&6&1&0&0&0\\
8&1&16&551&78&58&26&7&1&0&0\\
9&1&20&89&160&154&90&34&8&1&0\\
10&1&25&136&305&365&269&131&43&9&1\\ \hline
\end{tabular}
\caption{The first values of $b_{k,n}$, the number of parallelogram polyominoes having $k$ columns and area $n$.}
\label{t1}

\end{table}

Summing the values of each line we obtain the number of parallelogram polyominoes with $n$ cells, which corresponds to sequence \href{https://oeis.org/A006958}{A006958} of the OEIS \cite{OEIS}.

The Proposition \ref{p13} can be generalized to the $3$-dimensional case by the following way.

\begin{prop}
Let $p_{n_1,n_2,...,n_k}$ be the number of parallelogram polycubes of width $k$ and whose $i^{th}$ plateau has a volume $n_i$ with $i=1,...,k$.
$$p_{n_1,n_2,...,n_k}=\sum_{v_1|n_1,...,v_k|n_k}\prod_{i=1}^{k-1}min(v_i,v_{i+1})min(\frac{n_i}{v_i},\frac{n_{i+1}}{v_{i+1}}).$$

\end{prop}
\begin{proof}
First for a fixed volume of a plateau $n_i$, there are exactly $\tau(n_i)$ possible plateaus, where $\tau(n_i)$ is the number of divisor of $n_i$, an exemple is shown in Fig. \ref{f23bis}. Giving a plateau of volume $n_i$ and height $v_i$, we deduce that its depth is $\frac{n_i}{v_i}$. To glue two plateaus the same way as for polyominoes, the difference is that for the $2$-dimensional case we glued according to the height but in the case of polycubes we do it according to the height and depth. Therefore for a plateau $i$ of volume $n_i$ and height $v_i$ and a plateau $i+1$ of volume $n_{i+1}$ and height $v_{i+1}$, we have exactly $min(v_i,v_{i+1})min(\frac{n_i}{v_i},\frac{n_{i+1}}{v_{i+1}})$ gluing. An example is shown in Fig. \ref{f23}. Morevover by summing for all possible heights of each plateau, we get the formula.
\end{proof}
\begin{figure}[H]
\centering
\begin{tikzpicture}[scale=0.25]
	\foreach \x/\y/\z in {0/0/0,0/1/0,0/2/0,0/3/0,0/4/0,0/5/0,
	11/0/0,11/1/0,11/2/0,11/0/1,11/1/1,11/2/1,22/0/0,22/1/0,
	22/0/1,22/1/1,22/0/2,22/1/2,33/0/0,33/0/1,33/0/2,33/0/3,33/0/4,33/0/5
	}{
					\cube{\x}{\y}{\z};
				}

\end{tikzpicture}
\caption{All plateaus of volume $6$.}
\label{f23bis}
\end{figure}
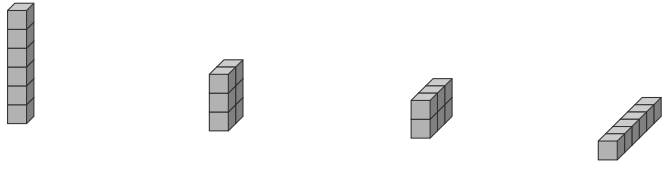

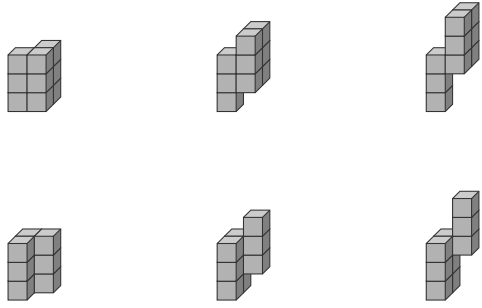
\begin{figure}[H]
\centering
\begin{tikzpicture}[scale=0.25]
	\foreach \x/\y/\z in {0/0/0,0/1/0,0/2/0,0/0/1,0/1/1,0/2/1,1/0/0,1/1/0,1/2/0,
	11/0/0,11/1/0,11/2/0,11/0/1,11/1/1,11/2/1,12/1/0,12/2/0,12/3/0,
	22/0/0,22/1/0,22/2/0,22/0/1,22/1/1,22/2/1,23/2/0,23/3/0,23/4/0,
	0/10/1,0/11/1,0/12/1,1/10/0,1/11/0,1/12/0,1/10/1,1/11/1,1/12/1,
	11/10/1,11/11/1,11/12/1,12/11/0,12/12/0,12/13/0,12/11/1,12/12/1,12/13/1,
	22/10/1,22/11/1,22/12/1,23/12/0,23/13/0,23/14/0,23/12/1,23/13/1,23/14/1}{
					\cube{\x}{\y}{\z};
				}

\end{tikzpicture}
\caption{Example of gluing two plateaus.}
\label{f23}
\end{figure}

Also, by summing for all partition of an integer $n$ into $k$ parts we get $c_{k,n}$ the number of parallelogram polycubes of width $k$ and volume $n$ in the following corollary.

\begin{cor}\label{cor1}
$$c_{k,n}=\sum_{\substack{n_1,n_2,...,n_k\in \mathbb{N}\\n_1+n_2+...+n_k=n}}p_{n_1,n_2,...,n_k}.$$
\end{cor}

The computing of the first values $n_1,n_2,...,n_k$ of $c_{k,n}$ the number of parallelogram polycubes of width $k$ and volume $n$ are given in the following table.

\begin{table}[H]
\begin{tabular}{|c|P{1cm}P{1cm}P{1cm}P{1cm}P{1cm}P{1cm}P{1cm}P{1cm}P{1cm}P{1cm}|}\hline
$n \backslash k$  & 1 & 2& 3&4&5&6&7&8&9&10 \\ \hline
1&1&0&0&0&0&0&0&0&0&0\\
2&2&1&0&0&0&0&0&0&0&0\\
3&2&4&1&0&0&0&0&0&0&0\\
4&3&10&6&1&0&0&0&0&0&0\\
5&2&18&22&8&1&0&0&0&0&0\\
6&4&32&59&38&10&1&0&0&0&0\\
7&2&44&132&132&58&12&1&0&0&0\\
8&4&70&264&374&245&82&14&1&0&0\\
9&3&84&469&916&836&406&110&16&1&0\\
10&4&126&808&2015&2438&1614&623&142&18&1\\ \hline
\end{tabular}
\caption{The first values of $c_{k,n}$, the number of parallelogram polycubes of width $k$ and area $n$.}
\label{t2}

\end{table}

The values of the diagonals corresponds to the results found experimentally in \cite{champarnaud2013enumeration}.

\subsection{Dirichlet generating function}

Let $V_{k}(x_1,x_2,...,x_k)$ be the Dirichlet generating function of the parallelogram polyominoes of width $k$, where $x_i$ codes the area of the $i^{th}$ column where $1\leq i \leq k$, it is defined as following

$$V_{k}(x_1,x_2,...,x_k):=\sum_{n_1,n_2,...,n_k\geq 1}\frac{min(m_1,m_2)\cdots min(n_{k-1},n_k)}{n_1^{x_1}n_2^{x_2}\cdots n_k^{x_k}}$$ 
Before giving an algorithm to how express this generating function in terms of Multiple Zeta Function, we give the formulas for $k=1$, $k=2$,$k=3$ and $k=4$.

\begin{itemize}
\item For $k=1$,
$$V_1(x_1)=\sum_{n_1\geq 1}\frac{1}{n_1^{x_1}}=\zeta_1(x_1).$$
\item For $k=2$, 
\begin{align*}
V_2(x_1,x_2)=&\sum_{n_1,n_2\geq 1}\frac{min(n_1,n_2)}{n_1^{x_1}n_2^{x_2}}\\
=&\sum_{n_1>n_2\geq 1}\frac{1}{n_1^{x_1}n_2^{x_2-1}}+\sum_{n_2>n_1\geq 1}\frac{1}{n_1^{x_1-1}n_2^{x_2}}+\sum_{n_1=n_2\geq 1} \frac{1}{n_1^{x_1+x_2-1}}.
\end{align*}
So $$V_2(x_1,x_2)=\zeta_2(x_1,x_2-1)+\zeta_2(x_2,x_1-1)+\zeta_1(x_1+x_2-1).$$
\item For $k=3$
\begin{align*}
V_3(x_1,x_2,x_3)=&\sum_{n_1,n_2,n_3\geq 1}\frac{min(n_1,n_2)min(n_2,n_3)}{n_1^{x_1}n_2^{x_2}n_3^{x_3}}\\
=&\sum_{n_1>n_2>n_3\geq 1}\frac{1}{n_1^{x_1}n_2^{x_2-1}n_3^{x_3-1}}+\sum_{n_1>n_3>n_2\geq 1}\frac{1}{n_1^{x_1}n_2^{x_2-2}n_3^{x_3}}\\
+&\sum_{n_2>n_1>n_3\geq 1}\frac{1}{n_1^{x_1-1}n_2^{x_2}n_3^{x_3-1}}+\sum_{n_2>n_3>n_1\geq 1}\frac{1}{n_1^{x_1-1}n_2^{x_2}n_3^{x_3-1}}\\
+&\sum_{n_3>n_1>n_2\geq 1}\frac{1}{n_1^{x_1-1}n_2^{x_2-1}n_3^{x_3}}+\sum_{n_3>n_2>n_1\geq 1}\frac{1}{n_1^{x_1-1}n_2^{x_2-1}n_3^{x_3}}\\
+&\sum_{n_1=n_2>n_3\geq 1}\frac{1}{n_1^{x_1+x_2-1}n_3^{x_3-1}}+\sum_{n_1=n_3>n_2\geq 1}\frac{1}{n_1^{x_1+x_3}n_2^{x_2-2}}\\
+&\sum_{n_2=n_3>n_1\geq 1}\frac{1}{n_1^{x_1-1}n_2^{x_2+x_3-1}}+\sum_{n_3>n_1=n_2\geq 1}\frac{1}{n_1^{x_1+x_2-2}n_3^{x_3}}\\
+&\sum_{n_2>n_1=n_3\geq 1}\frac{1}{n_1^{x_1+x_3-2}n_2^{x_2}}+\sum_{n_1>n_2=n_3\geq 1}\frac{1}{n_1^{x_1}n_2^{x_2+x_3-2}}\\
+&\sum_{n_1=n_2=n_3\geq 1}\frac{1}{n_1^{x_1+x_2+x^3-2}}.
\end{align*}

So 
\begin{align*}
V_3(x_1,x_2,x_3)=&\zeta_3(x_1,x_2-1,x_3-1)+\zeta_3(x_1,x_3,x_2-2)+\zeta_3(x_2,x_1-1,x_3-1)\\
+&\zeta_3(x_2,x_3-1,x_1-1)+\zeta_3(x_3,x_1,x_2-2)+\zeta_3(x_3,x_2-1,x_1-1)\\
+&\zeta_2(x_1+x_2-1,x_3-1)+\zeta_2(x_1+x_2,x_3-2)+\zeta_2(x_2+x_3-1,x_1-1) \\
+& \zeta_2(x_3,x_1+x_2-2)+\zeta_2(x_2,x_1+x_2-2)+\zeta_2(x_1,x_2+x_3-2)\\
+&\zeta_1(x_1+x_2+x_3-2). 
\end{align*}

\item For $k=4$

\begin{align*}
V_4(x_1,x_2,x_3,x_4)=&\zeta_4(x_1,x_2-1,x_3-1,x_4-1)+\zeta_4(x_1,x_2-1,x_4,x_3-2)\\
+&\zeta_4(x_1,x_3,x_2-2,x_4-1)+\zeta_4(x_1,x_3,x_4-1,x_2-2)\\
+&\zeta_4(x_1,x_4,x_2-1,x_3-2)+\zeta_4(x_1,x_4,x_3-1,x_2-2)\\
+&\zeta_4(x_2,x_1-1,x_3-1,x_4-1)+\zeta_4(x_2,x_1-1,x_4,x_3-2)\\
+&\zeta_4(x_2,x_3-1,x_1-1,x_4-1)+\zeta_4(x_2-1,x_3-1,x_4-1,x_1-1)\\
+&\zeta_4(x_2,x_4,x_1-1,x_3-2)+\zeta_4(x_2,x_4,x_3-2,x_1-1)\\
+&\zeta_4(x_3,x_2-1,x_1-1,x_4-1)+\zeta_4(x_3,x_2-1,x_4-1,x_1-1)\\
+&\zeta_4(x_3,x_1,x_2-2,x_4-1)+\zeta_4(x_3,x_1,x_4-1,x_2-2)\\
+&\zeta_4(x_3,x_4-1,x_2-1,x_1-1)+\zeta_4(x_3,x_4-1,x_1,x_2-2)\\
+&\zeta_4(x_4,x_2,x_3-2,x_1-1)+\zeta_4(x_4,x_2,x_1-1,x_3-2)\\
+&\zeta_4(x_4,x_3-1,x_2-1,x_1-1)+\zeta_4(x_4,x_3-1,x_1,x_2-2)\\
+&\zeta_4(x_4,x_1,x_2-1,x_3-2)+\zeta_4(x_4,x_1,x_3-1,x_2-2)\\
+&\zeta_3(x_1+x_2-1,x_3-1,x_4-1)+\zeta_3(x_1+x_2-1,x_4,x_3-2)\\
+&\zeta_3(x_3,x_1+x_2-2,x_4-1)+\zeta_3(x_4,x_1+x_2-1,x_3-2)\\
+&\zeta_3(x_3,x_4-1,x_1+x_2-2)+\zeta_3(x_4,x_3-1,x_1+x_2-2)\\
+&\zeta_3(x_1+x_3,x_2-2,x_4-1)+\zeta_3(x_1+x_3,x_4-1,x_2-2)\\
+&\zeta_3(x_2,x_1+x_3-2,x_4-1)+\zeta_3(x_4,x_1+x_3-1,x_2-2)\\
+&\zeta_3(x_2,x_4,x_1+x_3-3)+\zeta_3(x_4,x_2,x_1+x_3-3)\\
+&\zeta_3(x_1+x_4,x_2-1,x_3-2)+\zeta_3(x_1+x_4,x_3-1,x_2-2)\\
+&\zeta_3(x_2,x_1+x_4-1,x_3-2)+\zeta_3(x_3,x_1+x_4-1,x_2-2)\\
+&\zeta_3(x_2,x_3-1,x_1+x_4-2)+\zeta_3(x_3,x_2-1,x_1+x_4-2)\\
+&\zeta_3(x_2+x_3-1,x_1-1,x_4-1)+\zeta_3(x_2+x_3-1,x_4-1,x_1-1)\\
+&\zeta_3(x_1,x_2+x_3-2,x_4-1)+\zeta_3(x_4,x_2+x_3-2,x_1-1)\\
+&\zeta_3(x_1,x_4,x_2+x_3-3)+\zeta_3(x_4,x_1,x_2+x_3-3)\\
+&\zeta_3(x_2+x_4,x_1-1,x_3-2)+\zeta_3(x_2+x_4,x_3-2,x_1-1)\\
+&\zeta_3(x_1,x_2+x_4-1,x_3-2)+\zeta_3(x_3,x_2+x_4-2,x_1-1)\\
+&\zeta_3(x_1,x_3,x_2+x_4-3)+\zeta_3(x_3,x_1,x_2+x_4-3)\\
+&\zeta_3(x_3+x_4-1,x_1,x_2-2)+\zeta_3(x_3+x_4-1,x_2-1,x_1-1)\\
+&\zeta_3(x_1,x_3+x_4-1,x_2-2)+\zeta_3(x_2,x_3+x_4-2,x_1-1)\\
+&\zeta_3(x_1,x_2-1,x_3+x_4-2)+\zeta_3(x_2,x_1-1,x_3+x_4-2)\\
+&\zeta_2(x_1+x_2-1,x_3+x_4-2)+\zeta_2(x_3+x_4-1,x_1+x_2-2)\\
+&\zeta_2(x_1+x_3,x_2+x_4-3)+\zeta_2(x_2+x_4,x_1+x_3-3)\\
+&\zeta_2(x_1+x_4,x_2+x_3-3)+\zeta_2(x_2+x_3-1,x_1+x_4-2)\\
+&\zeta_2(x_1+x_2+x_3-2,x_4-1)+\zeta_2(x_4,x_1+x_2+x_3-3)\\
+&\zeta_2(x_1+x_2+x_4-1,x_3-2)+\zeta_2(x_3,x_1+x_2+x_4-3)\\
+&\zeta_2(x_1+x_3+x_4-1,x_2-2)+\zeta_2(x_2,x_1+x_3+x_4-3)\\
+&\zeta_2(x_2+x_3+x_4-2,x_1-1)+\zeta_2(x_1,x_2+x_3+x_4-3)\\
+&\zeta_1(x_1+x_2+x_3+x_4-3). 
\end{align*}

\end{itemize}

Let $X_k$ be the set of variables
$$X_k=\{x_1,x_2,...,x_k\}.$$
For $k\geq 2$, we have the following theorem

\begin{theorem}\label{thp}
For $k\geq 1$,
$$V_k(x_1,x_2,...,x_k)=\sum_{\substack{S\in S(X)\\l=card(S)}}\zeta_{l}(e_1,e_2,...,e_l),$$
where $S=(S_1,S_2,...,S_l)$ and , for $1\leq i \leq l$, $e_i=\sum\limits_{\substack{j=1\\x_j\in S_i}}^{k}x_j-f_i$, with
$f_i=\sum_{j=1}^{k}f_{i,j}^{+}+f_{i,j}^{-}$, 

\[f_{i,j}^{+} = \left\{ 
\begin{array}{l l}
  1, & \quad \text{if $x_{j+1}\in S_t$, $1\leq t\leq i-1$}\\
  0, & \quad \text{otherwise.}\\ \end{array} \right. \]
and
\[f_{i,j}^{-} = \left\{ 
\begin{array}{l l}
  1, & \quad \text{if $x_{j-1}\in S_t$, $1\leq t\leq i$}\\
  0, & \quad \text{otherwise.}\\ \end{array} \right. \]
\end{theorem}

\begin{proof}
In a first time, let us decompose the sum to make appear all the $n_i's$ possible orders, using the relations $>$ and $=$. Note that the number of these decompositions is equal to Fubini numbers which corresponds to the sequence \href{https://oeis.org/A000670}{A000670} of the OEIS \cite{OEIS}.\\
Then to each decomposition we associate $S\in S(X)$.\\
For $1\leq i \leq k-1$ and $1 \leq r,t\leq l$, let $x_i\in S_r$ and $x_{i+1}\in S_t$ with $r>t$ (resp. $t<r$). This means that for the associated $n_i's$, $n_i>n_{i+1}$ (resp. $n_i<n_{i+1}$), so $min(n_i,n_{i+1})=n_i$ (resp. $min(n_i,n_{i+1})=n_{i+1}$). It implies that simplifying the fraction, we can decrease the power of $n_i$ (resp. $n_{i+1}$). Thus, in the denominator of the fraction we get $n_i^{x_i}n_{i+1}^{{x_{i+1}-1}}$ (resp. $n_i^{x_i-1}n_{i+1}^{x_{i+1}}$). So the variable in the Zeta function $e_r=x_i$ and $e_t=x_{i+1}-1$ (resp. $e_r=x_i-1$ and $e_t=x_{i+1}$).\\
If $x_i,x_{i+1}\in S_r$, then
%it implies that 
$n_i=n_{i+1}$. Using the same reasoning, we get in the Zeta Function $e_r=x_i+x_{i+1}-1$.\\
Applying the same reasoning for $x_i$ and $x_{i-1}$, we add $-1$ to the variable $e_t$ (resp. $e_{r}$). The variable $f_{i,t}$ is therefore defined to count $-1's$ for each $x_i$ in $S_t$. Finally if the $S_t$ contains more than one element, $e_t$ is the sum of its variables and the sum of their $f_{i,t}$.

\end{proof}

\noindent
The Algorithm \ref{algo} gives the Dirichlet generating function for a given width $k$.\\
Note that:
\begin{itemize}
\item $Orderedpartition(X)$ gives all the ordered partitions of the set $X$.
\item $Index()$ gives the index of an element, for example, $Index(x_i)=i$. 
\item $Append()$ add an element to a vector, example $(1,2,3).Append(4)=(1,2,3,4)$.
\end{itemize}

\begin{algorithm}
\caption{The Dirichlet generating function of parallelogram polycubes}
\label{algo}
        \SetKwInOut{Input}{input}
        \SetKwInOut{Output}{output}
        \Input{$k$: Width of the polyominoes}
        \Output{$V_k$: Dirichlet generating function of parallelogram polyominoes of with $k$.}
%%\SetKwBlock{Beginn}{beginn}{ende}
\Begin{
$X=\{\}$\\
$V_k=0$\\
\For{$i\gets 1$ \KwTo $k$}{
$X \gets X\cup \{x_i\}$}
$P \gets Orderedpartition(X)$\\
\For{$t\gets 1$ \KwTo $Card(P)$}{
$S\gets P[t]$\\
$\zeta\gets [ ]$\\
\For{$i\gets 1$ \KwTo $Card(S)$}{
$e_i \gets 0$\\
\For{$j\gets 1$ \KwTo $Card(S[i])$}{
$m \gets Index(S[i,j])$\\
$f1 \gets 0$\\
$f2 \gets 0$\\
$f \gets 0$\\
\For{$r\gets 1$ \KwTo $i-1$}{
\If{$x_{m-1}\in S_r$}{
$f1 \gets f1+1$
}
}
\For{$r\gets 1$ \KwTo $i$}
{\If{$x_{m+1}\in S_r$}{
$f2 \gets f2+1$
}
}
$f\gets f1+f2$\\
$e_i \gets e_i+S[i,j]-f$
}
$\zeta.Append(e_i)$
}
$V_k \gets V_k+\zeta$
}
}% end for begin
    \end{algorithm}

%%%Note that:e Given a set $Y$, $Y[i]$ means the $i^{th}$ element of $Y$. 

Let $P_k(x_1,x_2,...,x_k)$ be the Dirichlet generating function of parallelogram polycubes of width $k$ and where $x_i$ codes the volume of the $i^{th}$ plateau, with $1\leq i \leq k$.
$$P_k(x_1,x_2,...,x_k):=\sum_{n_1,n_2,...,n_k\geq 1}\frac{\sum_{v_1|n_1,...,v_k|n_k}\prod_{j=1}^{k-1}min(v_j,v_{j+1})min(\frac{n_j}{v_j},\frac{n_{j+1}}{v_{j+1}})}{n_1^{x_1}n_2^{x_2}\cdots n_k^{x_k}}.$$

This generating function can be expressed according to $V_k(x_1,x_2,...,x_k)$ in the following theorem.

\begin{theorem}\label{thPolD}
For $k\geq 1$,
$$P_k(x_1,x_2,...,x_k)=\big(V_k(x_1,x_2,...,x_k)\big)^2$$
\end{theorem}

\begin{proof}

\begin{align*}
P_k(x_1,x_2,...,x_k)&=\sum_{n_1,n_2,...,n_k\geq 1}\frac{\sum_{v_1|n_1,...,v_k|n_k}\prod_{j=1}^{k-1}min(v_j,v_{j+1})min(\frac{n_j}{v_j},\frac{n_{j+1}}{v_{j+1}})}{n_1^{x_1}n_2^{x_2}\cdots n_k^{x_k}}\\
&=\sum_{n_1,n_2,...,n_k\geq 1}\frac{\prod_{j=1}^{k-1}min(n_j,n_{j+1})}{n_1^{x_1}n_1^{x_1}n_2^{x_2}\cdots n_k^{x_k}}\sum_{n_1,n_2,...,n_k\geq 1}\frac{\prod_{j=1}^{k-1}min(n_j,n_{j+1})}{n_1^{x_1}n_1^{x_1}n_2^{x_2}\cdots n_k^{x_k}}.\\
\end{align*}
Thus we get the result.
\end{proof}

\section{Enumeration according to the height, width and depth}\label{Sect04}

As seen in Section \ref{Sect03}, the enumeration of parallelogram polycubes is related to the case of polyominoes. In fact, by definition of parallelogram polycubes in \cite{carre2015enumeration}, the projections on $(0,\vec{i},\vec{j})$ and $(0,\vec{i},\vec{k})$ of a parallelogram polycube gives two parallelogram polyominoes, both have the same width as the polycube. Also, the height of the first polyomino is equal to the height of the polycube and the height of the second one is equal to the depth of the polycube. Moreover, from each each pair of parallelogram polyominoes having the same width, we can build a unique parallelogram polycube. The example in Fig. \ref{f14} illustrates the projections of the previous polycube.

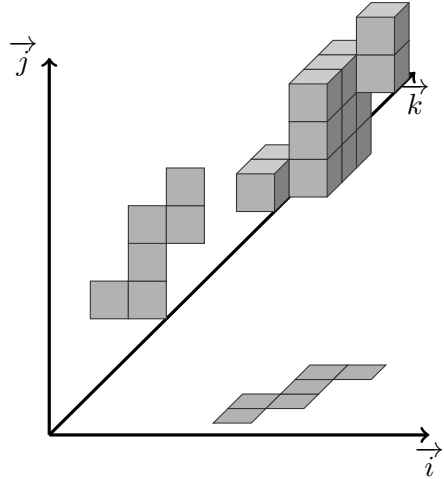
\begin{figure}[H]
\centering
\begin{tikzpicture}[scale=0.5]
\draw[very thick, ->] (2,0 ,13) -- (2,10,13) node[left]{$\overrightarrow{j}$};
\draw[very thick, ->] (2,0, 13) -- (2,0,-12) node[below]{$\overrightarrow{k}$};
\draw[very thick, ->] (2, 0,13) -- (12, 0,13) node[below]{$\overrightarrow{i}$};
	\foreach \x/\y/\z in {5/4/7,5/4/8,6/4/5,6/4/6,6/4/7,6/5/5,6/5/6,6/5/7,
								6/6/5,6/6/6,6/6/7,7/6/5,7/7/5}{
					\cube{\x}{\y}{\z};
				}
			\foreach \x/\y/\z in {0/0/5,1/0/5,1/1/5,1/2/5,2/2/5,2/3/5}{
					\freeza{\x}{\y}{\z};
				}		
				\foreach \x/\y/\z in {-1/-7/-8,-1/-7/-7,0/-7/-8,0/-7/-9,0/-7/-10,1/-7/-10}{
					\buu{\x}{\y}{\z};
				}
					
\end{tikzpicture}
\caption{Parallelogram polycube and its projections.}
\label{f14}
\end{figure}

In this section, we enumerate parallelogram polycubes according to the width, height and depth. We first start by giving in Lemma \ref{LPP}, the number of parallelogram polyominoes according to the width and height. Then from this result, we deduce the formula for the polycubes.  
\begin{lemma}\cite{DEUTSCH1999167}
\label{LPP}
Let $k,n \in \mathbb{N}$ and $g_{k,n}$ denote the number of parallelogram polyominoes of width $k$ and height $n$. Then,
$$g_{k,n}=\frac{1}{k+n-1}\binom{k+n-1}{k}\binom{k+n-1}{n}.$$
\end{lemma}
From this lemma, the following theorem is deduced.
\begin{theorem}
Let $s_{k,n,m}$ be the number of parallelogram polycubes, of width $k$, height $n$ and depth $m$. Then for $k,n,m\geq 1$, we have,
$$s_{k,n,m}=\frac{nm}{k^2(n+k-1)(m+k-1)}\binom{n+k-1}{k-1}^2\binom{m+k-1}{k-1}^2$$
\end{theorem}

\begin{proof}
%From the definition of the parallelogram polyominoes from each two 
As recalled,
for each pair of parallelogram polyominoes having the same width $k,$ we build a unique polycube of width $k$. If the height of the first one is $n$, then the height of the polycube is also $n$ and if the height of the second polyomino is $m$, it implies that the depth of the polycube $m$. From Lemma \ref{LPP}, the number of polyominoes of width $k$ and height $n$ is $\frac{1}{n+k-1}\binom{n+k-1}{k}\binom{n+k-1}{n}$, therefore by the rule of product, the formula is:
{\small
$$s_{k,n,m}=\frac{1}{(n+k-1)(m+k-1)}\binom{n+k-1}{n}\binom{n+k-1}{k}\binom{m+k-1}{m}\binom{m+k-1}{k}.$$
}%
Using the binomial identity \cite{quaintance2015combinatorial} 
{\small
$$\binom{n+k-1}{k}=\binom{n+k}{k}-\binom{n+k-1}{k-1},$$
}%
we obtain,
{\small
\begin{align*}
\frac{1}{(n+k-1)}\binom{n+k-1}{k}\binom{n+k-1}{k-1}&=\frac{1}{(n+k-1)}\Bigg[\binom{n+k}{k}-\binom{n+k-1}{k-1}\Bigg]\binom{n+k-1}{k-1}\\
&=\frac{1}{(n+k-1)}\Bigg[\frac{n+k}{k}\binom{n+k-1}{k-1}-\binom{n+k-1}{k-1}\Bigg]\binom{n+k-1}{k-1}\\
&=\frac{n}{k(n+k-1)}\binom{n+k-1}{k-1}^2.
\end{align*}
}%
By replacing in the formula, we get the result.
\end{proof}
Sequences of the OEIS \cite{OEIS} appear for some $k$, $n$ and $m$. The following table gives these cases.

\begin{table}[H]
\centering
\begin{tabular}{|c|c|}\hline
Sequence & Index OEIS\\ \hline
$s_{k,n,n}$& \href{https://oeis.org/A174158}{A174158}  \\ \hline
$s_{2,2,m}$& \href{https://oeis.org/A045943}{A045943}  \\ \hline
$s_{n,n,1}$& \href{https://oeis.org/A000891}{A000891}  \\ \hline
$a_l=\sum_{n+k=l}s_{k,n,n}$& \href{https://oeis.org/A319743}{A319743}  \\ \hline
\end{tabular}
\caption{Sequences of the OEIS that appear in the enumeration of parallelogram polycubes.}
\label{tseq}

\end{table}

For $k\geq 1$, let $$S_k(x,y)=\sum_{n,m\geq 1}p_{k,n,m}x^ny^m$$ be the generating function of parallelogram polycubes of width $k$ according to the height (coded by $x$) and the depth (coded by $y$).
To get the formula for the generating function, we use the following lemma.
\begin{lemma}\cite{DEUTSCH1999167}
Let $G_k(x)$ be the generating function of parallelogram polyominoes of width $k$ according to the height. Then for $k,m\geq 1$,
$$G_k(x)=\frac{B_{k-1}(x)}{(1-x)^{2k-1}},$$
where $B_{k}=\sum_{i=1}^{k-1}g_{i,k-1}x^{i}$ and $B_0=1$.
\end{lemma}
\begin{theorem}\label{ceTh}
For $k\geq 1$,
$$S_k(x,y)=\frac{B_{k-1}(x)B_{k-1}(y)}{((1-x)(1-y))^{2k-1}}.$$
\end{theorem}
\begin{proof}
\begin{align*}
S_k(x,y)&=\sum_{n,m\geq 1}s_{n,m,k}x^ny^m \\
&=\sum_{n,m\geq 1}g_{n,k}g_{m,k}x^ny^m \\
&=\bigg(\sum_{n\geq 1}g_{n,k}x^n\bigg)\bigg(\sum_{m\geq 1}g_{m,k}y^m \bigg)\\
&=G_k(x)G_k(y).
\end{align*}
Therefore we get the result.
\end{proof}

\section{Extension to any dimension}\label{Sect05}
In this section, the results found for polycubes will be extended for any dimension $d\geq 4$.

\subsection{Preliminaries}
In $\mathbb{Z}^d$, a \textit{polyhypercube} of dimension $d$, also called \textit{d-polycube}, is a finite union of cells (unit hypercubes), connected by their hypercubes of dimension $d-1$, and defined up to translation.

Let $(0, \vec{i_1}, \vec{i_2},..., \vec{i_d} )$ be an orthonormal coordinate system. The volume of polyhypercube is the number of its hypercubes. The \textit{width} is the difference between its greatest index and its smallest index according to $\vec{i_1}$. The \textit{$j^{th}$-height} is the difference between its greatest index and its smallest index according to $\vec{i_j}$, with $2\leq j \leq d$.\\
An elementary step is a positive move of one unit along the axis $\vec{i_j}$ with $1\leq j \leq d$. A polyhypercube is directed, if each cell can be reached from a distinguished cell called root, by a path only made by elementary steps.\\
The bottom (resp. top) according to $\vec{i_k}$, is the closet (resp. furthest) side of the hyperplan $(0, \vec{i_1},..., \vec{i_{k-1}},\vec{i_{k+1}},..., \vec{i_d} )$, with $k=2,...,d-1$. The bottom (resp. top) according to $\vec{i_d}$ is the closet (resp. furthest) side of the hyperplan $(0, \vec{i_1},..., \vec{i_{d-1}})$.
A \textit{stratum} is polyhypercube of width one. A \textit{plateau} is an hyperrectangular stratum without holes. A parallelogram polyhypercube is polyhypercubes whose bottoms and tops of its starata according to $\vec{i_k}$ form increasing sequences, $k=2,...,d$.\\
From this definitions we deduce that the projection of parallelogram polyhypercube on each plane gives a parallelogram polyomino.
Therefore, the results of the two previous sections can be extended to any dimension $d\geq 4$.
 
Let $P_{d,k}(x_1,x_2,...,x_k)$ be the generating function of parallelogram polhypercube of dimension $d$ and width $k$. Then, $P_{d,k}$ can be expressed according to $V_k$ the generating function of parallelogram polyominoes of width $k$.
 
\begin{theorem}
For $d\geq 4$ and $k\geq 1$,
$$P_{d,k}(x_1,x_2,...,x_k)=\big(V_k(x_1,x_2,...,x_k)\big)^{d-1}.$$
\end{theorem}

\begin{proof}
The proof is the same as the case of dimension $3$. It is deduced from the fact that the projection of a parallelogram polyhypercube on each plane gives a parallelogram polyomino. The volume of the $i^{th}$ hyperplateau is obtained by the product of the volumes of $i^{th}$ column of each polyomino obtained in the projection $1\leq i \leq k$.
\end{proof}

Let $s_{k,n_1,...,n_{d-1}}$ be the number of parallelogram polyhypercubes of dimension $d$, width $k$ and the $i^{th}$ height is equal to $n_i$, with $i=1...d-1$.
Using the same reasoning as for Theorem \ref{ceTh}, we obtain

\begin{theorem}
For $d\geq 4$, $k\geq 1$ and $n_i\geq 1$ with $i=1...d-1$,
$$s_{k,n_1,...,n_{d-1}}=\frac{1}{k^{d-1}}\prod_{i=1}^{d-1}\frac{n_i}{n_i+k-1}\binom{n_i+k-1}{k-1}^2.$$
\end{theorem}

%\begin{proof}
%The proof the same as the case $d=3$.
%\end{proof}

\bibliographystyle{acm}
\bibliography{Parl_Biblio}
\end{document}